\documentclass{article}
\usepackage{ijcai19}

\usepackage{times}
\usepackage{soul}
\usepackage{url}
\usepackage[utf8]{inputenc}
\usepackage[small]{caption}
\usepackage{graphicx}
\usepackage{amsmath}
\usepackage{booktabs}
\usepackage{algorithm}
\usepackage{algorithmic}
\usepackage{amssymb}
\usepackage{amsthm}
\usepackage{color}

\newtheorem{myTheorem}{Theorem}

\newtheorem{myDefinition}{Definition}

\newtheorem{myLemma}{Lemma}

\title{Pure Strategy Best Responses to Mixed Strategies in Repeated Games}

\author{
Shiheng Wang
 \And
 Fangzhen Lin
\affiliations
 Hong Kong University of Science and Technology
\emails
 \{swangbv, flin\}@cse.ust.hk
}

\begin{document}
\maketitle

\begin{abstract}
  Repeated games are difficult to analyze, especially when agents play mixed strategies.
  We study one-memory strategies in iterated prisoner's dilemma, then generalize the result to k-memory strategies in repeated games.
  Our result shows that there always exists a pure strategy best response, which can be computed with SMT or MDP solvers.
  However, there may not exist such pure strategy best response in multi-agent tournaments.
  All source code is released for verification(see additional files).
\end{abstract}

\section{Introduction}
    Repeated games are hard to analyze. It's well known that Nash equilibrium is the solution concept in one-shot games, whereas there exist infinite Nash equilibria in repeated games according to the folk theorem. In most repeated games, a player responses to the previous actions of the other player, thus the player's strategy need to be analysed with dynamic game theory~\cite{HandbookDynamic2018}. When the game is infinitely repeated, one cannot simply add up the payoff of each round, which in general will be infinite. Average rewards or discounted rewards are usually used to evaluate a strategy, and we focus on the former in this paper.

    In repeated games, a player's strategy is a function from histories of interactions to actions. Sometimes one restricts strategies to some specific forms, such as Turing machines~\cite{chen2015bounded,knoblauch1994computable,megiddo1986play}, finite automata~\cite{rubinstein1986finite,ben1990complexity,gilboa1988complexity,zuo2015optimal}, ones with limited memories ~\cite{hauert1997effects,lindgren1992evolutionary,chen2017k}, and other forms of bounded rationality~(e.g. \cite{osborne1994course,shoham2008multiagent}).

    A mixed strategy maps histories to a probability distribution of actions. \cite{press2012iterated} concluded that shortest memory sets the rule of the game, then \cite{chen2017k} pointed out the best response to a k-memory strategy in infinitely repeated games should also be k-memory. We further explore the best response to mixed strategies. Our result shows that there always exists a pure strategy best response in two-agent infinitely repeated games, but it's not the case in multi-agent tournaments.

    \begin{table}[t]
    \caption{Prisoner's Dilemma}\label{tab:PD}
      \centering
        \begin{tabular}{|c|c|c|}
          \hline
            &   c   & d \\
          \hline
          c & (R,R) & (S,T) \\
          \hline
          d & (T,S) & (P,P) \\
          \hline
        \end{tabular}
    \end{table}

     In the first place we study a concrete repeated game, namely the iterated prisoner's dilemma (IPD). It involves two agents playing repeatedly the Prisoner's Dilemma (PD) in table~\ref{tab:PD}. In the PD, each player can choose between Cooperate (c) and Defect (d). If both choose c, they receive a payoff of $R$ (rewards); If both choose d, they receive a payoff of $P$ (penalty); If one chooses c and the other d, the defector receives a payoff of $T$ (temptation to defect) and the cooperator receives a payoff of $S$ (sucker's payoff). The assumption is that $T>R>P>S$ and that $2*R > T + S$, which makes $(d,d)$ the only equilibrium in one-shot game, but cooperation provides more utility in the long run.

    When both agents take one-memory mixed strategies, the IPD can be modeled as a Markov chain, whose stationary distribution can be computed according to~\cite{press2012iterated} and the best response can be solved using an SMT solver like Z3~\cite{de2008z3}. We give a method to compute the best response and then summarize the best response to some popular strategies. In the meanwhile, our analysis explains the behavior of evolutionary agents. The solves a problem in~\cite{press2012iterated} as we can now formally prove the behaviour of these agents.

    Some of our results on one-memory strategies can be generalized to k-memory mixed strategies. In order to compute the best response to a completely mixed strategy, we build a Markov decision process (MDP) and compute its optimal policy. Later we prove that this is a communicating MDP, which has a pure optimal policy independent of the initial state.
    This shows that there always exists a pure strategy best response to any completely mixed strategy in a repeated game.

    However, this result does not hold for tournaments with more than two agents.
    We use a similar Markov chain model to compute the best strategy in a tournament of several one-memory agents, and show that the best response can not be a pure strategy.

    In this paper, when a strategy $\mathbf{q}$ is given, we calculate its best response $\mathbf{p}$.
    Most calculations or proofs are too laboursome to be done manually, so we show how state-of-art computer solvers assist us in solving this classic economic problem and bringing us new insights. In most cases, our program solved the problems in seconds. We release all our source code for verification.

    The rest of this paper is organized as follows. First we study one-memory strategies of the iterated prisoner's dilemma. Then we model k-memory problems into MDP and solve it with existing algorithms.
    Later we analyze a multi-agent tournament which serves as a counter example. Finally we give some discussion and concluding remarks.

\section{One-Memory Strategies in IPD}\label{sec:ipd}
\subsection{One-Memory Strategies}
In the section we consider a concrete example where both players play one-memory strategies in the iterated prisoner's dilemma (IPD).

The PD game in Table~\ref{tab:PD} is played by two players X and Y for infinite rounds.
The score (or payoff) of either player is calculated by average score in the limit (cf. \cite{shoham2008multiagent}).
Given an infinite sequence of scores $s_i^{(1)},s_i^{(2)},...$ for player $i \in \{X,Y\}$, the average score of $i$ is
\begin{equation}\label{eq:avarageScore}
    s_i = \lim_{k \to \infty } \frac{\sum_{j=1}^{k}s_i^{(j)}}{k}
\end{equation}

Although a player can have unlimited memory and decide what to do based on the entire history of the interactions so far, one-memory strategies base their response only on the outcome of the previous round. Press and Dyson \cite{press2012iterated} proved that shortest-memory player sets the rules of the game, in the sense that, for any strategy of the longer-memory player Y, X's score is exactly the same as if Y had played a certain shorter memory strategy, disregarding any history in excess of that shared with X. This conclusion enables us to focus on one-memory strategies.

A one-memory strategy consists of an initial state $p_0$ (the probability to cooperate in the first round) and a vector $\textbf{p}=(p_1,p_2,p_3,p_4)=(p_{cc},p_{cd},p_{dc},p_{dd})$ where $p_z$ is the probability of playing \emph{Cooperate} when the outcome $z$ occurred in the previous round.

If X uses the initial probability $p_0$ and strategy $\textbf{p}=(p_1,p_2,p_3,p_4)$, Y uses the initial probability $q_0$ and strategy $\textbf{q}=(q_1,q_2,q_3,q_4)$, then the probability distribution of the first iteration is
    $\textbf{v} ^1=(p_0q_0,p_0(1-q_0),(1-p_0)q_0,(1-p_0)(1-q_0))$ and the successive outcomes follow a Markov chain with transition matrix given by:
        \begin{equation}\label{MarkovMatrix}
         \mathbf{M} =
        \begin{pmatrix}
          p_1q_1 & p_1(1-q_1) & (1-p_1)q_1 & (1-p_1)(1-q_1) \\
          p_2q_3 & p_2(1-q_3) & (1-p_2)q_3 & (1-p_2)(1-q_3) \\
          p_3q_2 & p_3(1-q_2) & (1-p_3)q_2 & (1-p_3)(1-q_2) \\
          p_4q_4 & p_4(1-q_4) & (1-p_4)q_4 & (1-p_4)(1-q_4)
        \end{pmatrix}\nonumber
        \end{equation}
The probability distribution in r-th iteration $\mathbf{v}^r$ over the set of outcomes is a non-negative vector with unit sum, indexed by four states, $\mathbf{v}^r=(v_{cc}^r,v_{cd}^r,v_{dc}^r,v_{dd}^r)=(v_1^r,v_2^r,v_3^r,v_4^r)$.

Notice that we define outcome from each player's perspective, for example, if the outcome is $cd$ from X's perspective, it is $dc$ from Y's.
$\mathbf{v}$ is defined from player X's perspective, meaning that $v_{cd}$ refers to X plays C and Y plays D. If the previous outcome is $cd$, X's probability of cooperation is $p_2$ while Y's is $p_3$. See Table~\ref{tab:outcomeStrategy} for this correspondence.

Thus each entry of $\mathbf{M}$ represents the probability of transition between different states, which satisfies
\begin{equation}\label{eq:iteratev}
  \mathbf{M}\mathbf{v}^r=\mathbf{v}^{r+1}
\end{equation}
In accordance with \cite{akin2016iterated}, we will call $\mathbf{M}$ convergent when $\mathbf{M}$ has a unique stationary distribution of $\mathbf{v}$, which satisfies,
    $$\mathbf{M}\mathbf{v}=\mathbf{v}$$

\cite{press2012iterated} gave a determinant representation of player's payoff when it is calculated by the limit of average.
For an arbitrary four-vector $\mathbf{f}=(f_1,f_2,f_3,f_4)$ , let
    \begin{equation}\label{eq:Dpqf}
     D(\mathbf{p},\mathbf{q},\mathbf{f}) =
                                \begin{vmatrix}
                                    p_1q_1 - 1    & p_1-1          &  q_1-1        & f_1 \\
                                    p_2q_3        & p_2-1          &  q_3          & f_2 \\
                                    p_3q_2        & p_3            &  q_2-1        & f_3 \\
                                    p_4q_4        & p_4            &  q_4          & f_4 \\
                                \end{vmatrix}
    \end{equation}
    Then the average payoff of $s_X$ and $s_Y$ can be calculated with payoff vector $\mathbf{S_X}=(R,S,T,P)$ and $\mathbf{S_Y}=(R,T,S,P)$.
    \begin{equation}\label{eq:sXsY}
        s_{X} = \frac{D(\mathbf{p},\mathbf{q},\mathbf{S_X})}{D(\mathbf{p},\mathbf{q},\mathbf{1})}, \quad
        s_{Y} = \frac{D(\mathbf{p},\mathbf{q},\mathbf{S_Y})}{D(\mathbf{p},\mathbf{q},\mathbf{1})}
    \end{equation}

\begin{table}[t]
  \begin{center}
    \caption{Outcome and Strategy}\label{tab:outcomeStrategy}
    \begin{tabular}{|c|c|c|c|c|}
      \hline
      Outcome$^*$  & $cc$ & $cd$ & $dc$ & $dd$ \\ \hline
      Strategy of $\mathbf{p}$ & $p_1$ & $p_2$ & $p_3$ & $p_4$ \\ \hline
      Strategy of $\mathbf{q}$ & $q_1$ & $q_3$ & $q_2$ & $q_4$ \\
      \hline
    \end{tabular}\\
  \end{center}
  \begin{flushright}
    {\small * Outcome is defined from $\mathbf{p}$'s perspective}
  \end{flushright}
\end{table}

\subsection{Ergodic Markov Chain} \label{sec:ergodicOneMem}
    According to \cite{akin2016iterated}, the following statements on stationary distribution are equivalent.
    \begin{itemize}
      \item There is a unique stationary distribution $\mathbf{v}$ in accordance with $\mathbf{M}$.
      \item The stationary distribution $\mathbf{v}$ is independent of the initial distribution $\mathbf{v}^1$, and thus $p_0$ and $q_0$.
      \item There is no absorbing states (trapped states) in the Markov chain.
      \item $ D(\mathbf{p},\mathbf{q},\mathbf{1}) \neq 0 $
    \end{itemize}

    To avoid $ D(\mathbf{p},\mathbf{q},\mathbf{1}) = 0 $, we assume $\mathbf{q}$ plays a \textbf{completely mixed} strategy, that is, $q_i \in (0,1)$.
    For the convenience of defining best response, we assume $p_i\in [0,1]$ and $\mathbf{p} \neq (1,1,0,0)$ (namely strategy \emph{Repeat}). The last assumption makes sense because if $\mathbf{p}$ plays $Cooperate$ in the first round, strategy \emph{Repeat} is equivalent to \emph{Always Cooperate}, and otherwise to \emph{Always Defect}. Under these assumptions, we can prove
    $ D(\mathbf{p},\mathbf{q},\mathbf{1}) \neq 0 $ with the SMT solver Z3~\cite{de2008z3}.

    \begin{myTheorem}\label{thm:DNot0}
        Assume $p_i\in [0,1]$ and $q_i \in (0,1)$, and that $\mathbf{p} \neq (1,1,0,0)$, then $D(\mathbf{p},\mathbf{q},\mathbf{1}) < 0$.
    \end{myTheorem}

    \begin{proof}
        The negation of Theorem~\ref{thm:DNot0} is, $\exists \mathbf{p}, \mathbf{q}$
        \begin{align}\label{eq:prfDNot0}
            & \big ( \bigwedge_{i=1,2,3,4} 0 \leq p_i \leq 1 \big ) \quad \land \quad \big ( \bigwedge_{j=1,2,3,4} 0 < q_j < 1  \big ) \nonumber \\
            & \land \quad \big ( \mathbf{p} \neq (1,1,0,0) \big)\quad \land \quad \big ( D(\mathbf{p},\mathbf{q},\mathbf{1}) \geq 0 \big )
        \end{align}

        \emph{Z3} returns ``unsatisfiable'' to (\ref{eq:prfDNot0}), meaning that its value is always \emph{False}. Thus Theorem~\ref{thm:DNot0} is always \emph{True}.
    \end{proof}
    In practice we accelerate this proof with domain knowledge. We first prove that $D(\mathbf{p},\mathbf{q},\mathbf{1})$ is monotonic to $p_i$ by calculating derivatives(see also proof of Thm~\ref{thm:monotone}), then calculate extrema by letting $p_i \in \{0,1\}$. Since all extrema are less than zero, the determinant is less than zero.

\subsection{Symbolic Calculation}
     \begin{myTheorem}[Monotonicity]\label{thm:monotone}
         Suppose $p_i\in [0,1]$ and $q_i \in (0,1)$, $\mathbf{p} \neq (1,1,0,0)$, $s_X, s_Y$ are defined in equation (\ref{eq:sXsY}), and $z \in \{p_1,...,p_4,q_1,...,q_4\}$. \\
         When all variables in $\{p_1,...,p_4,q_1,...,q_4\}\backslash \{z\}$ are fixed, $s_X$ is monotonic to $z$.
     \end{myTheorem}
     \begin{proof}
        According to Theorem~\ref{thm:DNot0}, $D(\mathbf{p},\mathbf{q},\mathbf{1}) \neq 0$.
        $$ s_{X} = \frac{D(\mathbf{p},\mathbf{q},\mathbf{S_X})}{D(\mathbf{p},\mathbf{q},\mathbf{1})} $$
        For simplicity, let
        $$ M =D(\mathbf{p},\mathbf{q},\mathbf{S_X}), N = D(\mathbf{p},\mathbf{q},\mathbf{1}) $$
        Then
        $$ s_X = {M} \big / {N} $$
        Without loss of generality, get partial derivative to $p_1$,
        \begin{equation}\label{derivative}
        \frac{\partial s_X}{\partial p_1} = \frac{\partial M \big / \partial p_1 \cdot N  - \partial N \big / \partial p_1 \cdot M}{N^2}
        \end{equation}
        Let
        $$ U =  \partial M \big / \partial p_1 \cdot N  - \partial N \big / \partial p_1 \cdot M $$
        Again, get partial derivative of $U$ to $p_1$, $${\partial U} \big / {\partial p_1} = 0 $$
        In Eq. (\ref{derivative}), since $N^2 > 0$, and the numerator is not a function of $p_1$, we can draw the conclusion that, when $\{p_2,p_3,p_4,q_1,q_2,q_3,q_4\}$ are fixed, $s_X$ is monotonic in terms of $p_1$.

        Similarly, $s_X$ is also monotonic to any variable in $p_i,q_i$ when other variables are fixed.
     \end{proof}

     As symbolic calculations of determinant are laboursome, we use a \emph{Python} library \emph{SymPy}~\cite{sympy2017} and double checked with \emph{MATLAB}.

     Given a fixed completely mixed strategy $\mathbf{q}$, we want to compute its best response $\mathbf{p}$. Theorem~\ref{thm:monotone} implies that there exists a pure one-memory strategy best response. We only need to enumerate all pure strategy $\mathbf{p}$, and one of them must be a best response to $\mathbf{q}$.

     \begin{myTheorem}
        There always exists an one-memory pure strategy best response to completely mixed one-memory strategies.
     \end{myTheorem}

     \cite{press2012iterated} did some experiments to show that evolutionary players who evolve to get better payoff will finally reach the optimal reward but they didn't prove it analytically. Actually this property follows from Theorem~\ref{thm:monotone}. No matter where a strategy begins, it will finally evolve to the best response. In other words, there is no local maxima.

     We are now able to compute all extreme values of $s_X$ by letting $p_i = 0$ or $p_i = 1$. The result is denoted as $F_k$, where
     $ k = \sum_{i=1}^{4} p_i * 2 ^ {4-i} $. For instance, when $\mathbf{p}=(1,1,1,1)$,
     $$ F_{15} = \frac{-R*q_3 + S*q_1 - S}{q_1 - q_3 - 1} $$
     Some other $F_k$ has a much more complex expression. We'll not show them due to limited space but readers can refer to our source code for details.
     Notice that $F_{12}$ is not considered as it corresponds to \emph{Repeat} strategy, $\mathbf{p}=(1,1,0,0)$.
     Some formulas turn out to have the same value, that is,
        $$ F_0 = F_4 = F_8,\quad F_{13} = F_{14} = F_{15} $$
     The corresponding strategies always receive the same payoff whatever the co-player's strategy is, we call such strategies equivalent.
    \begin{myTheorem}[Equivalent]
        When playing agaisnt completely mixed strategies, the following strategies are equivalent.
    \begin{align}
      (0,0,0,0) \equiv (0,1,0,0) \equiv (1,0,0,0) \nonumber \\
      (1,1,0,1) \equiv (1,1,1,0) \equiv (1,1,1,1) \nonumber
    \end{align}
    \end{myTheorem}
    Duplicates will not be considered in the rest of this paper, and the set of distinct expressions is defined as,
    \begin{equation}\label{eq:distinctset}
        \mathcal{F} = \{F_0, F_1, F_2, F_3, F_5, F_6, F_7, F_9, F_{10}, F_{11}, F_{15} \}
    \end{equation}

    \subsection{Theorem Discovery}
    With concrete values of (R,S,T,P), when strategy $\mathbf{q}$ is given, we can simply compute all values in $\mathcal{F}$ so that the strategy corresponding to the largest value is a best response. Our experiments show that concrete values of (R,S,T,P) really matter.
    When the concrete setting of the IPD varies, the best responses to some strategies also change. But some general theorems are discovered with SMT solver Z3.

    The knowledge base is defined according to the constraints of the IPD and mixed one-memory strategies. In this section, $\supset$ means ``logical imply'' , $\equiv$ means ``logical equivalent'' and $\lnot$ means ``logical not''. Theorems are proved by refutation.
    \begin{myDefinition}[Knowledge Base]
    \begin{align}\label{eq:KB}
        \mathcal{KB} \equiv \quad & \big ( \bigwedge_{i=1,2,3,4} 0 \leq p_i \leq 1 \big ) \land  \big ( \bigwedge_{i=1,2,3,4} 0< q_i < 1 \big ) \nonumber \\
        & \land (T>R>P>S) \land (2R > T+S)
    \end{align}
    \end{myDefinition}

    If a strategy is less cooperative under a better situation, there is no reason to cooperate with him.
    As $T>R>P>S$ is the setting of the prisoner's dilemma, column player's preference is  $CD \succ CC \succ DD \succ DC$. If his strategy $\mathbf{q}$ satisfies $q_3 \leq q_1 \leq q_4 \leq q_2$, then there is no chance of reciprocity. See Table~\ref{tab:outcomeStrategy} for this correspondence.

    \begin{myTheorem}\label{thm:ungrateful}
        If strategy $\mathbf{q}=\{q_1,q_2,q_3,q_4\}$ satisfies $q_3 \leq q_1 \leq q_4 \leq q_2$, then \emph{Always Defect} is its best response .
    \end{myTheorem}
    \begin{proof}
        First, we give a formal representation of this theorem. $\forall \mathbf{p}, \mathbf{q}, R,S,T,P$,
        \begin{equation}\label{eq:bestLEQ}
            \mathcal{KB} \land (q_3 \leq q_1 \leq q_4 \leq q_2) \supset \bigwedge_{F' \in \mathcal{F}\backslash \{ F_{0} \} } F_{0} > F'
        \end{equation}
        Eq. (\ref{eq:bestLEQ}) is equivalent to, $\forall \mathbf{p}, \mathbf{q}, R,S,T,P$,
        \begin{equation}\label{eq:bestLEQ2}
         \lnot \big ( \mathcal{KB} \land (q_3 \leq q_1 \leq q_4 \leq q_2) \big ) \lor \bigwedge_{F' \in \mathcal{F}\backslash \{ F_{0} \} } F_{0} > F'
        \end{equation}
        This theorem can be proved by conclusion refutation, and the negation of Eq. (\ref{eq:bestLEQ2}) is, $\exists \mathbf{p},\mathbf{q}, R,S,T,P$,
        \begin{equation}\label{eq:bestLEQ3}
         \big ( \mathcal{KB} \land (q_3 \leq q_1 \leq q_4 \leq q_2) \big ) \land \bigvee_{F' \in \mathcal{F}\backslash \{ F_{0} \} } F_{0} \leq F'
        \end{equation}
        which is equivalent to,
        \begin{equation}\label{eq:bestLEQ4}
         \bigvee_{F' \in \mathcal{F}\backslash \{ F_{0} \}} \big ( \mathcal{KB} \land (q_3 \leq q_1 \leq q_4 \leq q_2)  \land  F_{0} \leq F' \big )
        \end{equation}
         Z3 returns ``unsatisfiable'' to equation (\ref{eq:bestLEQ4}), implying that its negation(i.e. Theorem~\ref{thm:ungrateful}) is always true.
    \end{proof}

    A natural corollary of Theorem \ref{thm:ungrateful} is,  if $q_1 = q_2 = q_3 = q_4$, then \emph{Always Defect} is a best response.
    One famous example of such strategy is \emph{Random}, i.e. $\mathbf{q}=(0.5,0.5,0.5,0.5)$.
    This corollary illustrates that if some player is indifferent to the outcome of previous round, there is no reason to cooperate with him.
    More discussion on this property and backward inducation can be found in Section~\ref{sec:backinduction}.

    Some other interesting theorems are discovered on the neighbor of \emph{Mischief} strategy defined by~\cite{press2012iterated}. Such strategy is also called equalizer.
    We only show the theorems together with their formal representations as the proofs are very similar.

    \begin{myDefinition}[MisChief]\label{def:mischief}
        A MisChief strategy is an one-memory strategy $\mathbf{q}=(q_1,\overline{q_2},\overline{q_3},q_4)$ defined as,
        \begin{align}\label{eq:mistort}
         \mathcal{MC} \equiv \quad  &   \overline{q_2} = \frac{{q_1}(T-P)-(1+{q_4})(T-R)}{R-P}  \quad \land \nonumber \\
                        &   \overline{q_3} = \frac{(1-{q_1})(P-S)+{q_4}(R-S)}{R-P}
        \end{align}
    \end{myDefinition}

    \begin{myTheorem}[Mischief]\label{thm:mischief}
        While playing with MisChief strategy $\mathbf{q}=(q_1,q_2,q_3,q_4)$, every strategy of $\mathbf{p}$ receives the same average payoff. Formally, $\forall \mathbf{p}, \mathbf{q}, R,S,T,P$
        \begin{equation}\label{thm:mischief1}
            \mathcal{KB} \land \mathcal{MC} \land (q_2 = \overline{q_2}) \land (q_3 = \overline{q_3}) \supset \bigwedge_{F_i, F_j \in \mathcal{F}} F_i = F_j \nonumber
        \end{equation}
    \end{myTheorem}

    \begin{myTheorem}[MisTort]\label{thm:mistort}
        While playing with strategy $\mathbf{q}=(q_1,q_2,q_3,q_4)$ s.t. $q_2 < \overline{q_2}$, $q_3 = \overline{q_3}$, Always Cooperate is a best response. Formally, $\forall \mathbf{p}, \mathbf{q}, R,S,T,P$,
        \begin{equation}\label{eq:mistort}
         \mathcal{KB} \land \mathcal{MC} \land (q_2 < \overline{q_2}) \land (q_3 = \overline{q_3})   \supset \bigwedge_{F' \in \mathcal{F}\backslash \{F_{15} \} } F_{15} > F' \nonumber
        \end{equation}
    \end{myTheorem}

    \begin{myTheorem}
        While playing with strategy $\mathbf{q}=(q_1,q_2,q_3,q_4)$ s.t. $q_2 = \overline{q_2}$,$q_3 > \overline{q_3}$, Always Defect is a best response.
        Formally, $\forall \mathbf{p}, \mathbf{q}, R,S,T,P$,
        \begin{equation}\label{eq:thmmistort}
         \mathcal{KB} \land \mathcal{MC} \land (q_2 = \overline{q_2}) \land (q_3 > \overline{q_3})   \supset \bigwedge_{F' \in \mathcal{F}\backslash \{F_{0} \} } F_{0} > F' \nonumber
        \end{equation}
    \end{myTheorem}

    \begin{figure}[t]
      \caption{MisTort Strategy $\mathbf{q}=(0.9, 0.5, 0.2, 0.1)$}\label{fig:mistort}
      \centering
      \includegraphics[width=\linewidth]{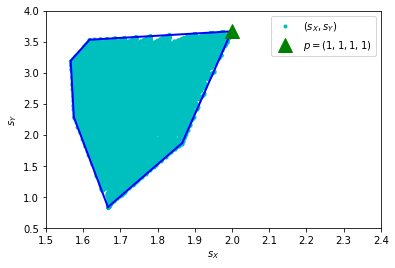}
    \end{figure}

    We call the strategy in Theorem~\ref{thm:mistort} \emph{MisTort} because it's on the neighbor of \emph{MisChief} strategy and shares the property of extortionate strategies in~\cite{press2012iterated}, to which the best response is \emph{Always Cooperate}. Due to similar property, readers may guess that one may be the a subclass of the other. However, there turns out to be no intersection between \emph{MisTort} and extortionate strategies.

    \begin{myTheorem}\label{thm:noIntersection}
        There is no intersection between extortionate strategy and MisTort strategy.
    \end{myTheorem}
    \begin{proof}
        Extortionate strategy($\mathcal{EX}$) is defined in~\cite{press2012iterated},
            \begin{align}
            \mathcal{EX} \equiv  \quad & q_1 = 1-\phi(\chi-1)\frac{R-P}{P-S}  \nonumber \\
            \land \quad & q_2 = 1-\phi(1+\chi\frac{T-P}{P-S}) \nonumber \\
            \land \quad & q_3 = \phi(\chi+\frac{T-P}{P-S}) \quad \land  \quad q_4 = 0 \nonumber \\
            \land \quad & 0 < \phi \leq \frac{(P-S)}{(P-S)+\chi (T-P)} \quad \land \quad \chi > 1 \nonumber
            \end{align}
        The negation of Theorem~\ref{thm:noIntersection} is, $\exists \mathbf{q}$,
        \begin{align}
           & (\bigwedge_{i=1,2,3,4} 0 \leq q_i \leq 1) \land \mathcal{EX}\land \mathcal{MC} \land (q_2 < \overline{q_2}) \nonumber \\
           & \land (q_3 = \overline{q_3}) \land (T>R>P>S) \land (2R > T+S) \label{eq:isIntersection}
        \end{align}
        Z3 returns ``unsat'', meaning  Theorem~\ref{thm:noIntersection} is always true.
    \end{proof}

    One specific instance of \emph{MisTort} strategy can be calculated by letting $q_1 = 0.9, q_4=0.1$ and $(R,S,T,P) = (3,0,5,1)$. From Eq. (\ref{eq:mistort}) we can calculate out $q_2<0.7, q_3=0.2$. A possible \emph{MisTort} strategy is, $\mathbf{q} = ( 0.9, 0.5, 0.2, 0.1 ) $.
    We enumerate a number of strategy $\mathbf{p}$ to play against this $\mathbf{q}=(0.9, 0.5, 0.2, 0.1)$, and the average payoff of both players is shown in Figure~\ref{fig:mistort}. All pairs of payoffs $(s_X, s_Y)$ form a compact convex hull, and the case that $\mathbf{p} = (1,1,1,1)$ is receives highest payoff.
    The figure implies that \emph{Always Cooperate} is a best response to $\mathbf{q}$, leading to a payoff where $s_X = 2.0, s_Y = 3.67$.

\section{Repeated Games and MDP}
\subsection{Repeated Game and K-memory Strategy}
    A finite, 2-person normal form game is a tuple $(N, A, u)$, where
    \begin{itemize}
      \item $N=\{1,2\}$ is the set of two players.
      \item $A = A_1 \times A_2$, where $A_i$ is a finite set of actions available to player $i$. Each vector $a = (a_1,a_2) \in A$ is called an action profile (or outcome).
      \item $u = (u_1,u_2)$, where $u_i : A \mapsto \mathbb{R}$ is a real-valued utility (or payoff) function for player i.
    \end{itemize}

    This stage game is played for infinite rounds. A player's payoff is defined as the average payoff of the stage game in the limit~\cite{shoham2008multiagent}.
    Similar to Eq.~(\ref{eq:avarageScore}), given an infinite sequence of payoffs $u_i^{(1)}$, $u_i^{(2)}$,... for player i, the average payoff of $i$ is
    $ \lim_{k\to\infty} {\sum_{j=1}^{k}u_i^{(j)}} \big / {k} $.

    A k-memory mixed strategy of player $i$ is a function that maps k action profiles to a distribution on actions,
    \begin{equation}\label{eq:kmemorystrategy}
        p_i: A^k  \mapsto \Delta(A_i)
    \end{equation}
    Given previous k actions profiles $s \in A^k$, the probability of taking action $a_i \in A_i$ is written as $p_i(s,a_i)$.
    In the start of the game, there aren't k action profiles, in this case we fill empty action profile with action $c$. Our conclusions don't rely on the outcomes of the first k rounds,
    see Thm~\ref{thm:communicating}.

\subsection{Markov Decision Process}\label{sec:mdp}
    Now we compute the best response to a k-memory strategy. The first question is how much memory is required to be a best response. According to~\cite{press2012iterated}, shortest memory sets the rule of the game, so that it needs to be at most k-memory. On the other hand, any strategy with less than k-memory can be represented by a k-memory strategy. Then we know k-memory strategy should be exactly enough. This problem is also discussed in~\cite{chen2017k}.

    Both players take k-memory strategies. Assume player 2 plays a completely mixed strategy $p_2$, namely $\forall s\in S, a_2 \in A_2, p_2(s,a_2) > 0 $, then from player 1's point of view, his action and payoff can be modeled as the following MDP.
    \begin{itemize}
      \item Actions: same as $A_1$ in the repeated game.
      \item States: $S = A^k $, the set of vectors of $k$ action profiles, indexed by 1,2,...,k, from farthest to most recent. Any $s_{1:k} \in S$ refers to k continuous action profiles. When there comes a new action profile $(a_1, a_2)$, it goes to a new state $s'=s_{2:k}+(a_1,a_2)$, which drops the farthest action profile $s_1$ and appends $(a_1,a_2)$.
      \item Transition Function. When player 1 takes action $a_1$, the states transits from state $s$ to $s'=s_{2:k}+(a_1,a_2)$ with probability
      $ T(s,a_1,s')= p_2(s,a_2) $. At every state $s$, for any action $a_1$, $\sum_{s'}^{}T(s,a_1,s') = 1$.
      \item Reward $R(s, a_1) = u_1(s_k), \forall a_1 \in A_1 $, meaning player 1's utility of the most recent action profile.
    \end{itemize}

    A mixed policy $\pi$ assigns a probability distribution of actions to each state. Formally, $\pi: S  \mapsto \Delta( A_1)$. For each $s \in S$, the probability of taking action $a_1 \in A_1$ is written as $\pi(s,a_1) $.
    A completely mixed policy is $\pi(s,a_1) > 0 , \forall s\in S, a_1 \in A_1$. $\Pi^+$ is the set of all completely mixed policies, which is a subset of all policies $\Pi$.
    A pure policy is that, for each $s\in S$, there is exactly one $a_1 \in A_1$ satisfying $\pi(s,a_1)=1$.
    From the definition of state $S$ and strategy in Eq.~(\ref{eq:kmemorystrategy}), we can easily see that every policy $\pi$ corresponds to a strategy $p_1$.

    Player 1 walks among states for infinite rounds, starting from any state $s\in S$. At each round, it takes an action according to current state and its policy $\pi$, then it goes into a new state. This agent receives a sequence of rewards, namely $R^1,R^2, ..., R^j, ...$, and the final score $ \rho^{\pi}(s)$ is calculated by the expectation of average, formally,
    $$ \rho^{\pi}(s)   = \lim_{k \to \infty } E^{\pi}\big \{ \frac{\sum_{j=0}^{k} R^j }{k}  \big \} $$

    Our goal is to find an optimal policy $\pi^*$ that maximizes the expectation of average payoff, $ \forall \pi, \rho^{\pi^*}(s) \geq  \rho^{\pi}(s) $. We will show that the initial state $s$ doesn't matter, for our MDP.

\subsection{Solving Average-Payoff MDPs}
    There is a detailed explanation on average-payoff MDPs in textbook~\cite{puterman2014markov}.
    Our results in this part are mainly based on propositions and theorems in~\cite{filar1988communicating}, which are introduced as lemmas.
    To be specific, Lemma~\ref{lem:purepolicy} is from paragraph 2, section 1; Definition~\ref{def:communicating} is from Definition 1.1,
    and Lemma~\ref{lem:communicating} is from Theorem 2.1 in~\cite{filar1988communicating}.

    \begin{myLemma}\label{lem:purepolicy}
        While computing optimal policy, it is sufficient to consider pure policies.
    \end{myLemma}

    A policy $\pi$ induces a Markov chain on states $S$ with transition matrix $\mathbf{M}(\pi)$, whose entries $\mathbf{M}_{st}(\pi)$ denote the probability of transition from state s to state t when policy $\pi$ is followed. For $\tau$ being a nonnegative integer, $\mathbf{M}^{\tau}(\pi)$ denotes the $\tau$-th power of square matrix $\mathbf{M}(\pi)$, and $\mathbf{M}^{\tau}_{st}(\pi)$ denotes an entry in it.

    \begin{myDefinition}\label{def:communicating}
        An MDP is communicating if, for every pair of states $s,t \in S$, there exists a pure policy $\pi$ and an integer $\tau \geq 1$ such that $\mathbf{M}^{\tau}_{st}(\pi)$ is strictly positive.
    \end{myDefinition}

    \begin{myLemma}\label{lem:communicating}
        Let $\Pi^+$ be the set of completely mixed policies. The following two conditions are equivalent.
        \begin{itemize}
          \item An MDP is communicating.
          \item Every policy $\pi^+ \in \Pi^+$ induces an irreducible $\mathbf{M}(\pi^+)$.
        \end{itemize}
    \end{myLemma}

    Since we have assumed that player 2 takes a completely mixed strategy $p_2$, for every $\pi^+ \in \Pi^+$, the induced matrix $\mathbf{M}(\pi^+)$ is irreducible.
    In other words, when both players take completely mixed strategies, there is no transient state in the corresponding Markov chain, because it's possible to append any action profile to any state.
    \begin{myTheorem}\label{thm:communicating}
        The MDP in section~\ref{sec:mdp} is communicating. Its optimal policy can be calculated by linear programming, which is independent of starting state.
        If $\pi^*$ is an optimal policy starting from state $s$, then it is also an optimal policy starting from any other state $s'$.
    \end{myTheorem}

    Then there is a theorem in repeated games.
    \begin{myTheorem}\label{thm:communicating}
        There always exists a k-memory pure strategy best response to k-memory completely mixed strategy in infinitely repeated games, which is independent of the initial k outcomes.
    \end{myTheorem}

    Such pure strategy best response can be computed with existing MDP solvers, e.g. MDPtoolbox~\cite{chades2014mdptoolbox}.
    As an example, we compute the best response to Stochastic Tif-for-2-Tats(\emph{STF2T}). The game is IPD in Table~\ref{tab:PD} with (R,S,T,P)=(3,0,5,1). \emph{STF2T} cooperates with probability 0.1 when the other player defects for continuous two rounds, and cooperates with probability 0.9 otherwise.
    MDPtoolbox solves this model, the best response is to play $c$ and $d$ alternatively, and the average payoff is 2.67. See our source code for details.

\section{Discussion}
\subsection{Multi-agent Tournament}\label{sec:multiagent}
    In the previous sections we have discussed how to calculate the best response in two-agent repeated games. Our result mainly relies on the fact that there always exists a pure strategy best response. Now we consider the best response in multi-agent tournaments, which may not be a pure strategy.

    Due to the complexity of symbolic calculations, we conduct experiments of a multi-agent IPD tournament instead of giving an analytical proof.
    Consider a tournament of eleven one-memory agents, namely one $\mathbf{p}$, nine $\mathbf{q}$ and one $\mathbf{u}$. Suppose $(R,S,T,P) = (3,0,5,1)$, $\mathbf{q} = ( 0.9, 0.5, 0.2, 0.1 ) $ and
    $\mathbf{u}=( 0.4, 0.8, 0.2, 0.6 )$, we want to compute the optimal $\mathbf{p}$ that maximizes his payoff. Strategy $\mathbf{p}$ plays with every $\mathbf{q}$ and $\mathbf{u}$ respectively and the final score takes the average of all games, which can be calculated according to Eq.~(\ref{eq:sXsY}),
    \begin{equation}\label{eq:tournament}
        s_{\mathbf{p}} = 0.9 * \frac{D(\mathbf{p},\mathbf{q},\mathbf{S_X})}{D(\mathbf{p},\mathbf{q},\mathbf{1})}
        + 0.1 * \frac{D(\mathbf{p},\mathbf{u},\mathbf{S_X})}{D(\mathbf{p},\mathbf{u},\mathbf{1})}
    \end{equation}
    In two-agent games, there always exists a pure strategy best response. As is shown in Theorem~\ref{thm:mistort} and Theorem~\ref{thm:ungrateful}, the best response to $\mathbf{q}$ is \emph{Always Cooperate}, while the best response to $\mathbf{u}$ is \emph{Always Defect}. Now we consider whether there is a pure strategy one-memory best response $\mathbf{p}^*$ in this tournament .

    First we let $\mathbf{p}$ be a pure strategy and compute the values of $s_{\mathbf{p}}$. There are $2^4 -1 = 15$ pure strategies, where \emph{Repeat} strategy $\mathbf{p}=(1,1,0,0)$ is excluded. Among them, strategy \emph{Tit-for-Tat} $\mathbf{p}=(1,0,1,0)$ receives highest payoff of 1.90. However, after we consider mixed strategies, we found that strategy $\mathbf{p}=(1,0.9, 0 ,0.1)$ reaches a higher payoff of 2.02.

    This specific example shows that there may not exist a pure strategy best response in general multi-agent tournaments.

\subsection{Limitation of \emph{Z3}}
    Although \emph{Z3} is effective in solving linear constraints and successfully prove our theorems, it sometimes fails to give a solution in several hours when the constraint is a non-linear combination of several variables. We take comprehensive measures to overcome this challenge. (1) \textbf{Avoid quantifiers}. A mixture of $\forall$ and $\exists$ prevents us from efficient proof. We avoid using both quantifiers and manually simplify some formulas. (2) \textbf{Simplify formulas}. We take advantage of domain knowledge to remove fractions, such as in the proof of Theorem~\ref{thm:DNot0} we first prove monotonicity and prove all extreme values are less than zero.
    (3) \textbf{Break down.} We solve one clause of conjunction or disjunction at one time. Usually a theorems requires every clause has the same value of {\it true} or {\it false}, such as Eq. (\ref{eq:bestLEQ4}).
    (4) \textbf{Remove variables}.
    When SMT cannot solve a theorem, we have to replace some variables with concrete values to get some conjectures or a weaker theorem. In section~\ref{sec:multiagent} we compute a counter-example instead of prove it analytically.

\subsection{Backward Induction}\label{sec:backinduction}
    There has been many debates on backward induction \cite{kreps1982rational}\cite{binmore1997rationality}. For finite n round IPD, playing defect is the best strategy in the one-shot game at n-th round because there is no further possibility of reciprocity. Since both rational players will play defect in the n-th round, there is no reason to cooperate in the (n-1)-th round. By backward induction, \emph{Always Defect} is the only equilibrium of finite iterated prisoner's dilemma.

    To explain the emergence of cooperation, it is usually assumed that the game is played for infinite rounds, or an unknown number of rounds\cite{axelrod1981evolution}. Although this assumption invalidates backward induction, there still lacks an explanation of cooperation. Actually, when people conduct backward induction, they implicitly assume that the outcome of previous round has no effect on the decision of current round. According to the corollary of Theorem \ref{thm:ungrateful}, when someone is indifferent to the outcome of previous round, there is no reason to cooperate with him. Therefore, we hold the view that the emergence of cooperation cannot be explained without taking previous history into consideration.

\section{Conclusion}
    In this paper we compute best responses to mixed strategies in repeated games. Our main result shows that there always exists a pure strategy best response in two-agent repeated games. Based on this result, we analyze one-memory strategies, give the method to compute best response, and discover new theorems in the iterated prisoner's dilemma. The work enhances our comprehension of the IPD and explains the evolutionary behavior left over in~\cite{press2012iterated}.

    We generalize this result to the best response to k-memory strategies. Such problem is modeled as MDP and solved with existing solvers.
    In a multi-agent tournament, however, there may not exist a pure strategy best response. As a result, computing the best strategy in a tournament
    where an agent should take the same strategy to all other agents
    is still an open question.

    As most calculations and proofs are conducted by computer programs, we release all our source code for verification. Source code is uploaded as additional files.

\bibliographystyle{named}
\bibliography{RepeatedGames,AlgrithmicGameTheory}

\begin{thebibliography}{}

\bibitem[\protect\citeauthoryear{Akin}{2016}]{akin2016iterated}
Ethan Akin.
\newblock The iterated prisoner’s dilemma: good strategies and their
  dynamics.
\newblock {\em Ergodic Theory, Advances in Dynamical Systems}, pages 77--107,
  2016.

\bibitem[\protect\citeauthoryear{Axelrod and
  Hamilton}{1981}]{axelrod1981evolution}
Robert Axelrod and William~Donald Hamilton.
\newblock The evolution of cooperation.
\newblock {\em science}, 211(4489):1390--1396, 1981.

\bibitem[\protect\citeauthoryear{Ben-Porath}{1990}]{ben1990complexity}
Elchanan Ben-Porath.
\newblock The complexity of computing a best response automaton in repeated
  games with mixed strategies.
\newblock {\em Games and Economic Behavior}, 2(1):1--12, 1990.

\bibitem[\protect\citeauthoryear{Binmore}{1997}]{binmore1997rationality}
Ken Binmore.
\newblock Rationality and backward induction.
\newblock {\em Journal of Economic Methodology}, 4(1):23--41, 1997.

\bibitem[\protect\citeauthoryear{Chad{\`e}s \bgroup \em et al.\egroup
  }{2014}]{chades2014mdptoolbox}
Iadine Chad{\`e}s, Guillaume Chapron, Marie-Jos{\'e}e Cros, Fr{\'e}d{\'e}rick
  Garcia, and R{\'e}gis Sabbadin.
\newblock Mdptoolbox: a multi-platform toolbox to solve stochastic dynamic
  programming problems.
\newblock {\em Ecography}, 37(9):916--920, 2014.

\bibitem[\protect\citeauthoryear{Chen and Tang}{2015}]{chen2015bounded}
Lijie Chen and Pingzhong Tang.
\newblock Bounded rationality of restricted turing machines.
\newblock In {\em Proceedings of the 2015 International Conference on
  Autonomous Agents and Multiagent Systems}, pages 1673--1674. International
  Foundation for Autonomous Agents and Multiagent Systems, 2015.

\bibitem[\protect\citeauthoryear{Chen \bgroup \em et al.\egroup
  }{2017}]{chen2017k}
Lijie Chen, Fangzhen Lin, Pingzhong Tang, Kangning Wang, Ruosong Wang, and
  Shiheng Wang.
\newblock K-memory strategies in repeated games.
\newblock In {\em Proceedings of the 16th Conference on Autonomous Agents and
  MultiAgent Systems}, pages 1493--1498. International Foundation for
  Autonomous Agents and Multiagent Systems, 2017.

\bibitem[\protect\citeauthoryear{De~Moura and Bj{\o}rner}{2008}]{de2008z3}
Leonardo De~Moura and Nikolaj Bj{\o}rner.
\newblock Z3: An efficient smt solver.
\newblock In {\em International conference on Tools and Algorithms for the
  Construction and Analysis of Systems}, pages 337--340. Springer, 2008.

\bibitem[\protect\citeauthoryear{Filar and
  Schultz}{1988}]{filar1988communicating}
Jerzy~A Filar and Todd~A Schultz.
\newblock Communicating mdps: equivalence and lp properties.
\newblock {\em Operations Research Letters}, 7(6):303--307, 1988.

\bibitem[\protect\citeauthoryear{Gilboa}{1988}]{gilboa1988complexity}
Itzhak Gilboa.
\newblock The complexity of computing best-response automata in repeated games.
\newblock {\em Journal of economic theory}, 45(2):342--352, 1988.

\bibitem[\protect\citeauthoryear{Han}{2018}]{HandbookDynamic2018}
{\em Handbook of dynamic game theory}.
\newblock Springer, Cham, Switzerland, 2018.

\bibitem[\protect\citeauthoryear{Hauert and Schuster}{1997}]{hauert1997effects}
CH~Hauert and Heinz~Georg Schuster.
\newblock Effects of increasing the number of players and memory size in the
  iterated prisoner's dilemma: a numerical approach.
\newblock {\em Proceedings of the Royal Society of London B: Biological
  Sciences}, 264(1381):513--519, 1997.

\bibitem[\protect\citeauthoryear{Knoblauch}{1994}]{knoblauch1994computable}
Vicki Knoblauch.
\newblock Computable strategies for repeated prisoner's dilemma.
\newblock {\em Games and Economic Behavior}, 7(3):381--389, 1994.

\bibitem[\protect\citeauthoryear{Kreps \bgroup \em et al.\egroup
  }{1982}]{kreps1982rational}
David~M Kreps, Paul Milgrom, John Roberts, and Robert Wilson.
\newblock Rational cooperation in the finitely repeated prisoners' dilemma.
\newblock {\em Journal of Economic theory}, 27(2):245--252, 1982.

\bibitem[\protect\citeauthoryear{Lindgren}{1992}]{lindgren1992evolutionary}
Kristian Lindgren.
\newblock Evolutionary phenomena in simple dynamics.
\newblock In {\em Artificial life II}, pages 295--312, 1992.

\bibitem[\protect\citeauthoryear{Megiddo and Wigderson}{1986}]{megiddo1986play}
Nimrod Megiddo and Avi Wigderson.
\newblock On play by means of computing machines: preliminary version.
\newblock In {\em Proceedings of the 1986 Conference on Theoretical aspects of
  reasoning about knowledge}, pages 259--274. Morgan Kaufmann Publishers Inc.,
  1986.

\bibitem[\protect\citeauthoryear{Meurer \bgroup \em et al.\egroup
  }{2017}]{sympy2017}
Aaron Meurer, Christopher~P. Smith, Mateusz Paprocki, Ond\v{r}ej
  \v{C}ert\'{i}k, Sergey~B. Kirpichev, Matthew Rocklin, AMiT Kumar, Sergiu
  Ivanov, Jason~K. Moore, Sartaj Singh, Thilina Rathnayake, Sean Vig, Brian~E.
  Granger, Richard~P. Muller, Francesco Bonazzi, Harsh Gupta, Shivam Vats,
  Fredrik Johansson, Fabian Pedregosa, Matthew~J. Curry, Andy~R. Terrel,
  \v{S}t\v{e}p\'{a}n Rou\v{c}ka, Ashutosh Saboo, Isuru Fernando, Sumith Kulal,
  Robert Cimrman, and Anthony Scopatz.
\newblock Sympy: symbolic computing in python.
\newblock {\em PeerJ Computer Science}, 3:e103, January 2017.

\bibitem[\protect\citeauthoryear{Osborne and
  Rubinstein}{1994}]{osborne1994course}
Martin~J Osborne and Ariel Rubinstein.
\newblock {\em A course in game theory}.
\newblock MIT press, 1994.

\bibitem[\protect\citeauthoryear{Press and Dyson}{2012}]{press2012iterated}
William~H Press and Freeman~J Dyson.
\newblock Iterated prisoner's dilemma contains strategies that dominate any
  evolutionary opponent.
\newblock {\em Proceedings of the National Academy of Sciences},
  109(26):10409--10413, 2012.

\bibitem[\protect\citeauthoryear{Puterman}{2014}]{puterman2014markov}
Martin~L Puterman.
\newblock {\em Markov decision processes: discrete stochastic dynamic
  programming}.
\newblock John Wiley \& Sons, 2014.

\bibitem[\protect\citeauthoryear{Rubinstein}{1986}]{rubinstein1986finite}
Ariel Rubinstein.
\newblock Finite automata play the repeated prisoner's dilemma.
\newblock {\em Journal of economic theory}, 39(1):83--96, 1986.

\bibitem[\protect\citeauthoryear{Shoham and
  Leyton-Brown}{2008}]{shoham2008multiagent}
Yoav Shoham and Kevin Leyton-Brown.
\newblock {\em Multiagent systems: Algorithmic, game-theoretic, and logical
  foundations}.
\newblock Cambridge University Press, 2008.

\bibitem[\protect\citeauthoryear{Zuo and Tang}{2015}]{zuo2015optimal}
Song Zuo and Pingzhong Tang.
\newblock Optimal machine strategies to commit to in two-person repeated games.
\newblock In {\em AAAI}, pages 1071--1078, 2015.

\end{thebibliography}

\end{document}